\documentclass[runningheads]{llncs}
\usepackage{graphicx}
\usepackage{tikz,tikz-cd}
\usepackage{amsmath,amssymb,amsthm}
\newtheorem{thm}{Theorem}
\newtheorem{prop}[thm]{Proposition}

\newtheorem{lem}[thm]{Lemma}

\newtheorem{dfn}[thm]{Definition}

\newcommand{\id}{\text{id}}
\newcommand{\op}[1]{{#1}^\mathrm{op}}

\newcommand*{\End}[1][]{\mathop{\text{End}_{#1}}}
\newcommand{\LMod}[1]{{_{#1}\text{Mod}}}
\newcommand{\RMod}[1]{\text{Mod}_{#1}}
\newcommand{\BiMod}[2]{{_{#1}\text{Mod}_{#2}}}
\DeclareMathOperator{\Hom}{Hom}
\DeclareMathOperator{\powSet}{\mathcal{P}}

\begin{document}
\title{Morita Rigidity for Kleene Algebras\thanks{The author gratefully acknowledges the support of Cornell University and Dexter Kozen.}}
\titlerunning{Morita Rigidity}
% If the paper title is too long for the running head, you can set
% an abbreviated paper title here
%
\author{Luke Serafin\inst{1}\orcidID{0009-0004-5189-1429}}
\authorrunning{L. Serafin}
% First names are abbreviated in the running head.
% If there are more than two authors, 'et al.' is used.
%
\institute{Cornell University, Ithaca NY 14850, USA \\
\email{lss255@cornell.edu}}
\maketitle              % typeset the header of the contribution
\begin{abstract}
We introduce Morita equivalence to the study of Kleene algebras and modules.
Classical characterizations of Morita-equivalent semirings such as having equivalent categories of modules
and one semiring being a full idempotent corner of a finite matrix semiring over the other transfer easily to this new context.
We also observe that Morita equivalence can be applied to extending and restricting scalars in Lindenbaum--Tarski algebras of propositional dynamic logics,
and that an appropriately-defined category of generalized finite automata satisfies the property that when the underlying Kleene algebras are Morita equivalent then the corresponding categories of generalized automata are equivalent.
But the signature results which we obtain relate to rigidity of Kleene algebras in this context.
Specifically, every semiring that is Morita-equivalent to the semiring reduct of a Kleene algebra has a unique expansion to a Kleene algebra, and semiring bimodules witnessing Morita-equivalence of semirings which are reducts of
Kleene algebras are already Kleene bimodules.
%But the signature results which we obtain are strong forms of rigidity for Kleene algebras, Kleene lattices, and action lattices.
%Specifically, if any one of these algebras has its semiring reduct Morita equivalent to a semiring,
%then there is an expansion of that semiring to an algebra of the same type as the original algebra which is Morita-equivalent to the original algebra, via bimodules for the full structure of the original algebra.
%Such strong rigidity is not completely general, because we observe that action algebras do not exhibit this strong form of rigidity.

\end{abstract}
\section{Introduction}

Modules and the notion of Morita equivalence are foundational to the classical study of rings~\cite{lam-lectures},
and we assume the reader has some familiarity with the classical theory of modules though not necessarily of Morita equivalence.

The theory of modules has previously been extended to the context of semirings (for example see~\cite{golan-semirings}), though some care is needed in choosing a tensor product for the modules with which we are concerned.
The tensor product used in~\cite{golan-semirings} is always cancellative, and is therefore trivial for modules over additively idempotent semirings.
The alternative construction of~\cite{katsov-injective-envelopes,katsov-nam-morita} has the properties which we need, and we give an account of this construction and its adjunction with a Hom functor
formed by homomorphisms of right semiring modules.
One contribution of this work is to extend to the setting of Kleene modules, which are semiring modules which additionally behave as expected with respect to the asterate of the Kleene algebra of scalars.
Although an earlier version of this paper claimed to verify a tensor-hom adjunction in the category of Kleene modules, the large language model GPT-6.1 Astra uncovered an issue where some homomorphism modules did not in fact
satisfy the Kleene module axioms, and moreover found a counterexample indicating that there is no tensor-hom adjunction as claimed.
This paper includes a corrected version, found by means of a collaborative effort of the human author and the LLM.
We also present a variety of instances of Kleene modules, the most important being one-sided endomorphism semirings of finite free modules (including $n \times n$ matrices)
and the Lindenbaum--Tarski algebra of propositional dynamic logic.

Morita equivalence for rings is obtained by replacing the morphisms of the category of rings by bimodules, and considering the induced notion of isomorphism.
More specifically, the morphisms from a ring $R$ to a ring $S$ are the $(R, S)$-bimodules, and composition of bimodules is given by tensor product.
The author has undertaken to study Morita equivalence in the context of Kleene algebras, leading to the discovery of a number of results and directions for further study.
Several standard results carry over, such as characterizations of Morita equivalence in terms of categorical equivalence of module categories, and in terms of matrix algebras defined in terms of full (multiplicative) idempotents.
One surprising fact is that if the semiring reducts of two Kleene algebras $K$ and $S$ are Morita-equivalent, then in fact $K$ and $S$ are Morita-equivalent as Kleene algebras (using Kleene bimodules to witness the equivalence).
This is a rigidity theorem, saying that a weaker equivalence (Morita equivalence in the category of semirings) implies a stronger equivalence (Morita equivalence in the category of Kleene algebras) within a particular domain of
structures (here, Kleene algebras).
It is natural to ask whether the assumption that both semirings already carry Kleene algebra structures is necessary.
In theorem~\ref{KA-morita-expansion} we show that it is not: if a semiring $S$ is Morita-equivalent to the semiring reduct of a Kleene algebra $K$, then $S$ admits an asterate making it a Kleene algebra Morita-equivalent to $K$.
To prove this we work directly with the pairs of bimodules (one in each direction) that witness Morita equivalence, constructing the asterate on $S$ and showing that these bimodules already satisfy the Kleene module axioms.
We also prove that every idempotent corner $eAe$ of a Kleene algebra $A$ is itself a Kleene algebra with unit $e$ and asterate $a \mapsto ea^*e$; this gives a second description of the expansion through the
matrix-corner characterization of Morita equivalence.

Finally, in section~\ref{sec-projective-automata} we consider an application to categories of automata.
One might expect Morita-equivalent Kleene algebras to have equivalent categories of finite automata, but for labelled graphs and bisimulations even the choice of generators can affect the resulting category.
We therefore consider a different definition, using finitely generated projective modules to model states and specifying a transition endomorphism and an accepting submodule.
Theorem~\ref{thm-projective-automata-morita} shows that the resulting categories are Morita-invariant, both when morphisms are only required to preserve acceptance and also when they are additionally required to reflect it.

%The analogous rigidity results for Kleene lattices (here understood as Kleene algebras with finite meets) and action lattices turn out to be true as well.
%The idea of looking for rigidity results for Kleene lattices or action algebras was given to the author in the form of a question of Jose Meseguer posed at Logic and Applications 2025 in Dubrovník, Croatia.

\section{Semiring Preliminaries}
We review here the basics of semirings and semiring modules.
We use the term \emph{semiring module}, or often just \emph{module} for what the literature usually calls a \emph{semimodule}.
This is so that we can speak of ``Kleene modules'' rather than ``Kleene semimodules.''
For more concerning semirings see for example~\cite{golan-semirings} and~\cite{hebisch-weinert-semirings-cs}.
The tensor product of semiring modules used in~\cite{golan-semirings} is always cancellative and becomes trivial when one factor is additively idempotent.
For an explicit comparison with the tensor product satisfying the usual universal property, see~\cite[section~2]{banagl-tensor-semimodule}, where the base semiring is commutative.
We use the balanced tensor product of semiring modules with underlying commutative monoids as developed in~\cite{katsov-injective-envelopes};
its extension to bimodules over arbitrary semirings is recalled in~\cite[section~2]{katsov-nam-morita}.
Our tensor products are always taken over semiring reducts. We do not assume that arbitrary tensor products or quotients of Kleene modules are Kleene modules.
For related tensor constructions for Kleene algebras and idempotent structures see~\cite{hopkins-leiss,eklund-etal-semigroups-complete-lattices}.

\begin{dfn}
A (unital) \emph{semiring} is an algebraic structure $\langle S, +, \cdot, 0, 1 \rangle$ such that
\begin{enumerate}
\item $\langle S, +, 0 \rangle$ is a commutative monoid,
\item $\langle S, \cdot, 1 \rangle$ is a monoid (not necessarily commutative),
\item $s \cdot 0 = 0 \cdot s = 0$ for every $s \in S$, and
\item multiplication distributes over addition, meaning that for any $x, y, z \in S$, $x \cdot (y + z) = (x \cdot y) + (x \cdot z)$ and $(x + y) \cdot z = (x \cdot z) + (y \cdot z)$.
\end{enumerate}
We often abuse notation and denote the semiring $\langle S, +, \cdot, 0, 1 \rangle$ by $S$, the operations being understood.
Also, we generally denote multiplication $\cdot$ by juxtaposition, $xy = x \cdot y$, and regard multiplication as taking precedence over addition in expressions, so $xy + z$ is parsed as $(x \cdot y) + z$.
When multiple semirings are under consideration we may add subscripts to operations, e.g. $\alpha(r +_S s) \cdot_T t$.
\end{dfn}

As is common in ring theory, we sometimes need multiplication in a semiring to be computed in the opposite order, and so define the operation on semirings which reverses the order of multiplication.

\begin{dfn}
For $S$ a semiring, the \emph{opposite} semiring $\op{S}$ has the same underlying set, addition, and constants as $S$, but the multiplication $\cdot_{\op{S}}$ is defined by $s \cdot_{\op{S}} t = t \cdot_S s$ for $s, t \in S$.
\end{dfn}

Note that $(\op{S})^\mathrm{op} = S$ for any semiring $S$.
We next give the definition of a semiring module; see, for example,~\cite{golan-semirings}.

\begin{dfn}
Let $S$ be a semiring.
A \emph{left semiring module over $S$} (or simply a \emph{left $S$-module}, where it is understood that $S$ is a semiring)
is a commutative monoid $\langle M, +, 0 \rangle$ together with a map $S \times M \rightarrow M$, denoted by juxtaposition---so $(s,m) \mapsto sm$ for $s \in S$ and $m \in M$---satisfying the following
for any $s, t \in S$ and $m, n \in M$:
\begin{enumerate}
\item $s(m + n) = sm + sn$
\item $(s + t)m = sm + tm$
\item $s(tm) = (st)m$
\item $1m = m$
\item $s0_M = 0_M = 0_Sm$.
\end{enumerate}
Right $S$-modules are defined symmetrically, with the operation now a map $M \times S \rightarrow M$, written on the right: $(m, s) \mapsto ms$.

For semirings $S,T$, an \emph{$(S,T)$-bimodule} is a commutative monoid $\langle M, +, 0 \rangle$ equipped with both a left $S$-module map and a right $T$-module map, such that these two structures are compatible in the sense
that for any $s \in S$, $t \in T$, and $m \in M$, $s(mt) = (sm)t$.
As in the case of semirings, when we need to disambiguate in which module an operation is happening, we use subscripts, as in $s \cdot_N \alpha(m +_M n \cdot_M t)$ where $M$ is a right semiring module and $N$ a left semiring module
over a semiring $S$.
As in the literature on modules over rings, we sometimes use subscripts to the left or right of a symbol denoting a module over a semiring to indicate both which semiring acts on the module and from which side.
Thus if $S$, $T$ are semirings, ${_S}M$ denotes a left module over $S$, $N_T$ denotes a right module over $T$, and ${_S}B_T$ indicates an $(S, T)$-bimodule.
This notation is particularly helpful when working with direct sums and tensor products, as we shall later see.
\end{dfn}

To define tensor products we shall need the notion of a congruence relation.
This is a general notion from universal algebra (see for example~\cite{burris-sankappanavar}).
We view modules as algebraic structures with a potentially infinite number of unary operations of scalar multiplication, one for each scalar.

\begin{dfn}
A \emph{congruence} on an algebraic structure $\mathfrak{A}$ is an equivalence relation on the underlying set of $\mathfrak{A}$ which preserves all operations.
For example, if $\mathfrak{A}$ is a semiring and $\equiv$ is a congruence on $\mathfrak{A}$, then for $a, b, a', b'$ elements of the underlying set of $\mathfrak{A}$,
if $a \equiv a'$ and $b \equiv b'$, then $a + b \equiv a' + b'$ and $ab \equiv a'b'$.
\end{dfn}

It is a well-known fact from universal algebra (see~\cite{burris-sankappanavar} for a proof) that if $\equiv$ is a congruence on an algebraic structure $\mathfrak{A}$,
then the set of equivalence classes of $\equiv$ carries a structure of the same type as $\mathfrak{A}$, with operations computed by taking the equivalence class of the corresponding operation in $\mathfrak{A}$ performed on
representatives.
This is the \emph{quotient} of the algebra $\mathfrak{A}$ by the congruence $\equiv$.
Thus for example if $\equiv$ is a congruence on a semiring $S$ and we use brackets to denote $\equiv$-equivalence classes, then in the quotient of $S$ by $\equiv$ we have for $s, t \in S$ that
$[s] + [t] = [s + t]$ and $[s][t] = [st]$.
It is easy to use the stipulation that $\equiv$ is a congruence to prove that these operations are well-defined.

As in the case of semirings, we can define the opposite of a semiring module.

\begin{dfn}
For $S$ a semiring and $M$ a left semiring module over $S$, the \emph{opposite} module $\op{M}$ to $M$ is a right $\op{S}$-module with the action of $\op{S}$ given by $m \cdot_{\op{M}} s = s \cdot_M m$ for $s \in S$ and $m \in M$.
The opposite of a right $S$-module and the opposite of an $(S, T)$-bimodule for $T$ also a semiring are defined similarly, with the opposite of a right $S$-module being a left $\op{S}$-module and the opposite of an $(S, T)$-bimodule
being a $(\op{T}, \op{S})$-bimodule.
\end{dfn}

\begin{prop}
The opposite of a module over a semiring $S$ is a module of the appropriate type.
\end{prop}

\begin{proof}
Let $S$ and $T$ be semirings; we treat just the case of an $(S, T)$-bimodule, since the other cases are straightforward adaptations of the argument.
Fix an $(S, T)$-bimodule $M$ and note that the actions of $\op{T}$ and of $\op{S}$ on $\op{M}$ are given respectively by
\[ t \cdot_{\op{M}} m = m \cdot_M t \quad \text{and} \quad m \cdot_{\op{M}} s = s \cdot_M m \]
for $s \in S$, $t \in T$, and $m \in M$.
It is immediate that these actions are additive, preserve zero, and have the property that $1_{\op{S}}$ and $1_{\op{T}}$ act as identities.
It remains to verify that for $s, s' \in S$, $t, t' \in T$, and $m \in M$, we have the following:
\begin{equation}
\label{left-act} t \cdot_{\op{M}} (t' \cdot_{\op{M}} m) = (t \cdot_{\op{T}} t') \cdot_{\op{M}} m,
\end{equation}
\begin{equation} 
\label{right-act} (m \cdot_{\op{M}} s) \cdot_{\op{M}} s' = m \cdot_{\op{M}} (s \cdot_{\op{S}} s'), \quad \text{and}
\end{equation}
\begin{equation}
\label{act-commute} (t \cdot_{\op{M}} m) \cdot_{\op{M}} s = t \cdot_{\op{M}} (m \cdot_{\op{M}} s).
\end{equation}
For \eqref{left-act} we compute that
\begin{align*}
 t \cdot_{\op{M}} (t' \cdot_{\op{M}} m) &= t \cdot_{\op{M}} (m \cdot_M t') = (m \cdot_M t') \cdot_M t = m \cdot_M (t' \cdot_T t) = \\
 &= m \cdot_M (t \cdot_{\op{T}} t') = (t \cdot_{\op{T}} t') \cdot_{\op{M}} m.
\end{align*}
The case of \eqref{right-act} is symmetric.
For \eqref{act-commute} note that
\[ (t \cdot_{\op{M}} m) \cdot_{\op{M}} s = (mt) \cdot_{\op{M}} s = s(mt) = (sm)t = t \cdot_{\op{M}} (sm) = t \cdot_{\op{M}} (m \cdot_{\op{M}} s). \]
\end{proof}

\begin{prop}
For $S$ a semiring, $\op{S}$ has the structure of an $(S, S)$-bimodule (and hence also of either a left or a right semiring module over $S$).
\end{prop}

\begin{proof}
Use the underlying additive monoid of $S$ and define $a\cdot_L x=ax$ and $x\cdot_R b=xb$, with products on the right computed in $S$.
Both actions satisfy the module axioms, and they commute because $(ax)b=a(xb)$.
The additive monoids underlying $S$ and $\op{S}$ are identical, so this equips the underlying additive monoid of $\op{S}$ with the asserted bimodule structure.
No compatibility with multiplication in $\op{S}$ is asserted.
\end{proof}

\section{Kleene Algebras and Kleene Modules}

A semiring $S$ is \emph{additively idempotent} if and only if for every $x \in S$, $x + x = x$.
Additively idempotent semirings (and in fact idempotent commutative monoids more generally)\footnote{Commutative idempotent monoids are easily seen to be definitionally equivalent to join-semilattices with a least element}
come equipped with a partial order defined by $x \le y$ if and only if $x + y = y$.
It is easy to check that $x + y$ is then the least upper bound of $x$ and $y$ in the order.

\begin{dfn}
A \emph{Kleene algebra} $\langle K, +, \cdot, 0, 1, * \rangle$ is an additively idempotent semiring with an additional unary operation $*$ that satisfies the following axioms for any $x, y, z \in K$ (see~\cite{kozen-ka-cl-semirings,kozen-completeness}):
\begin{itemize}
\item $1 + xx^* = 1 + x^* x = x^*$,
\item if $x + yz \le z$ then $y^* x \le z$,
\item if $x + zy \le z$ then $xy^* \le z$.
\end{itemize}
\end{dfn}

Note that, since semirings are defined to be unital, if a semiring $S$ is additively idempotent then any left or right module over $S$ is also additively idempotent.
The argument for a left module $M$ over $S$ will be sufficient, because the case of a right module is symmetric.
For $m \in M$,
\[ m + m = 1m + 1m = (1 + 1)m = 1m = m. \]
Also note that by the symmetry of the conditions on the asterate, the opposite semiring of a Kleene algebra becomes a Kleene algebra when endowed with the same asterate.

\begin{dfn}
Fix a semiring $S$, and let $M$ and $N$ be left $S$-modules.
A map $\phi : M \rightarrow N$ is a \emph{homomorphism} of left $S$-modules if and only if for every $s \in S$ and $m, n \in M$,
\[ \phi(m + n) = \phi(m) + \phi(n) \quad \text{and} \quad \phi(sm) = s\phi(m); \]
the definitions of homomorphisms of right $S$-modules and of $(S, T)$-bimodules for $T$ a semiring are analogous.
The set of all $S$-module homomorphisms from $M$ to $N$ is denoted $\Hom(M, N)$.
Note that this set naturally comes equipped with pointwise addition; later we shall see how it can be endowed with module structure when $M$ and $N$ have additional structure.
When $M = N$ we write $\End(M)$ for $\Hom(M,M)$ and speak of endomorphisms of the semiring module $M$.
The category of left $S$-modules and homomorphisms\footnote{With functional composition as the composition operation.} is denoted $\LMod{S}$, the category of right $S$-modules and homomorphisms is denoted $\RMod{S}$, and the category of $(S,T)$-bimodules and homomorphisms is denoted $\BiMod{S}{T}$.
In all these categories isomorphism has its usual, categorical meaning (for which see e.g.~\cite{maclane-working}).
\end{dfn}

\begin{dfn}
Let $K$ be a Kleene algebra.
A \emph{left Kleene module} over $K$ is a semiring module $M$ over the semiring reduct of $K$ satisfying the following implication for $m \in M$ and $a \in K$:
\[ a m \le m \Rightarrow a^* m \le m. \]
A \emph{right Kleene module} over $K$ is similarly a right semiring module $M$ satisfying $ma \le m \Rightarrow ma^* \le m$ for $m \in M$ and $a \in K$.
A \emph{Kleene $(K, L)$-bimodule} is a semiring $(K, L)$-bimodule which is simultaneously a left Kleene module over $K$ and a right Kleene module over $L$.
\end{dfn}

When $K$ is a Kleene algebra, we reserve $\LMod{K}$ and $\RMod{K}$ for its categories of left and right Kleene modules, respectively, and write $K_{\mathrm{red}}$ for its semiring reduct when referring to the categories of all semiring modules.

Ehm, M\"oller, and Struth~\cite{ehm-etal-kleene-modules} study Kleene modules whose module sort is a Boolean algebra, with scalar actions abstracting relational image and preimage.
Here we do not require Boolean structure on the module; we use semiring modules satisfying the asterate axiom above.

For a fixed Kleene algebra $K$, regard multiplication by each scalar as a unary operation on the module.
The semiring module axioms are equations, and the additional asterate axiom is a quasi-identity (recall that $x\le y$ means $x+y=y$).
Consequently the left Kleene $K$-modules form a quasivariety, and the same argument applies to right Kleene $K$-modules.
It follows by universal algebra that there are free Kleene $K$-modules over arbitrary sets of generators (see for example~\cite{bergman-coproducts-varieties}).

\subsection{Examples}

\textbf{Algebras.} 
Any Kleene algebra $K$ is a $(K, K)$-bimodule over itself, and a $(2, 2)$-bimodule over the Kleene algebra $2$ with two elements.
Moreover it is an $(A, A)$-bimodule over any subalgebra $A$ of $K$.
In particular, if $K$ is a Kleene algebra with tests~(as defined in~\cite{kozen-frederick-KAT}; a KAT), and $B$ is its distinguished Boolean subalgebra, then $B$ is a Kleene subalgebra of $K$ and $K$ is a $(B, B)$-bimodule.\\

\textbf{Ideals.}
Let $I$ be a left ideal of a Kleene algebra $K$.
Then $I$ is a left Kleene module with action given by multiplication.
To see this, note that $I$ is clearly a left module over an additively idempotent semiring  under the action of multiplication, and that for $m \in I$, if $a m \le m$ then by the $*$-axioms for Kleene algebras, $a^* m \le m$.
Similarly, if $J$ is a right ideal of $K$ then $J$ is a right Kleene module, and if $J$ is a two-sided ideal then it is a Kleene $(K, K)$-bimodule.\\

\textbf{Quotients.}
Let $\equiv$ be a congruence
 of a Kleene algebra $K$, and set $Q = K / {\equiv}$.
If $Q$ is in fact a Kleene algebra, then $Q$ is a Kleene $K$-bimodule with action given by multiplication.
This follows from the fact that $\equiv$ is a congruence, and provides an example of a Kleene algebra construction that comes with a natural Kleene bimodule structure.\\

\textbf{Function modules.}
Let $K$ be a Kleene algebra, $B$ be a set, and consider the set $K^B$ of functions from $B$ to $K$ endowed with pointwise addition and zero as the constant-zero function.
Define left and right actions of $K$ on $K^B$ by $(am)(b) = a m(b)$ and $(ma)(b) = m(b) a$ for $a \in K$, $m \in K^B$.
It is immediate that these operations give a $(K, K)$-bimodule structure on $K^B$.
Moreover, if $B$ is finite, then $K^B$ has the usual finite basis given by characteristic functions for singletons in $B$, and hence is free as a left module and as a right module on these coordinate vectors. It need not be free as a bimodule.
One can also restrict from $K^B$ to functions supported on a set in some lattice-ideal 
 for $\powSet(B)$. 
It is easy to see that the finitely-supported\footnote{So, supported in the ideal of finite sets.} elements of $K^B$ form the free $K$-module with generators $B$.\\

\textbf{Endomorphism modules.}
\begin{dfn}\label{dfn-one-sided-hom}
For a right $K$-module $P$, its endomorphism semiring is
\[ \End[K](P)=\Hom_K(P,P), \]
with pointwise addition and multiplication given by composition.
Its right-module dual is $P^\vee=\Hom_K(P,K)$, a left $K$-module under $(k f)(p)=k f(p)$.
\end{dfn}

For columns $P=K^{n\times1}$, left matrix multiplication identifies $\End[K](P)$ with $M_n(K)$, and the dual is the row module $K^{1\times n}$.
Note that these identifications depend on taking homomorphisms of right modules; requiring homomorphisms to preserve both actions would give a different construction.
For example, $\End_{K\text{-}K}(K)$ is isomorphic to the center of $K$ (the set of elements commuting with all elements of $K$).
For a Kleene algebra $K$, the left action on $\Hom_K(P,K)$ satisfies the asterate axiom for Kleene modules pointwise in the codomain.\\

\textbf{Propositional Dynamic Logic.}
Fix sets $P$ of action atoms and $T$ of test atoms, respectively.\footnote{For details concerning propositional dynamic logic see for instance~\cite{kozen-parikh-pdl}.}
Let $K$ be the free Kleene algebra generated by $P$, and $B$ the free Boolean algebra generated by $T$.
The Lindenbaum--Tarski algebra $A$ of propositional dynamic logic with action atoms $P$ and test atoms $T$ can be viewed as a $(B, B)$-bimodule via the subalgebra bimodule structure, but also has the structure
of a left $K$-module via the action $pa = [\langle p \rangle \phi]$ where $p \in K$, $a \in A$, $a = [\phi]$, and brackets denote the equivalence class of a formula in the Lindenbaum--Tarski algebra.
This follows immediately from the axioms for propositional dynamic logic.\\

One construction of Kleene modules claimed by the author in the RAMICS version of this paper turns out not to work in general, with a counterexample found by GPT-6.\\

\textbf{A failure of Kleene asterate: One-sided homomorphism modules.}

For semirings $A,B,C$, a $(B,C)$-bimodule $N$, and an $(A,C)$-bimodule $P$, let $\Hom_C(N,P)$ denote the homomorphisms of right $C$-modules.
It is an $(A,B)$-bimodule with pointwise addition and actions~\cite[proposition~2.2]{katsov-nam-morita}
\[ (a f b)(n)=a f(bn). \]
To see that these actions are well-defined, note that right $C$-linearity follows from the commuting actions on $N$ and $P$.
Associativity of the actions and the fact that they commute follow directly from the displayed formula.
This gives a semiring bimodule, but as we shall see, the asterate axiom for Kleene modules need not hold even when the original bimodules are Kleene.

To illustrate why arbitrary homomorphism modules require care, let $R$ be the Kleene algebra of regular languages over $\{a\}$, $\mathbb B$ be the Boolean (and hence Kleene) algebra with two elements,
and let $h:R\to\mathbb B$ send finite languages to $0$ and infinite languages to $1$.
Note that this is a homomorphism of right $\mathbb B$-modules.
In $\Hom_{\mathbb B}(R,\mathbb B)$ we have that $h\cdot a=h$, but
\[ (h\cdot a^*)(\{\epsilon\})=1>0=h(\{\epsilon\}). \]
Thus this homomorphism module is not a right Kleene $R$-module.

%\subsubsection{Kleene Representation Theory}

%This can be viewed either as a broad class of examples (which in fact can be seen as encompassing all examples), or simply as another view of Kleene modules.
%In ring theory, representations are another view of modules and provide another fundamental perspective on rings.
%This carries over to Kleene algebras:
%\begin{dfn}
%Let $K$ be a Kleene algebra and $\langle M, +, 0 \rangle$ be a commutative idempotent monoid.
%A (left) \emph{representation} of $K$ over $M$ is a semiring homomorphism $\phi \mathrel{:} K \rightarrow \End(M)$ with the property that for any $m \in M$ and $a \in K$,
%if $\phi(a)(m) \le m$ then $\phi(a^*)(m) \le m$.
%\end{dfn}

%It is easy to see that this is just another formulation of Kleene modules.
%In particular, if two Kleene algebras have equivalent categories of modules, then they have the same representation theory.
%Though equivalent to the module perspective, the representation perspective for rings has led to a substantial body of work which is not as natural from the module perspective,
%and perhaps the representation perspective on Kleene modules will be of conceptual service as well.

\subsection{Tensor products of semiring modules}

The construction below is recalled in~\cite[section~2, pp.~447--448]{katsov-nam-morita}; proposition~2.1 there gives the outer bimodule actions and the tensor bifunctor.

\begin{dfn}\label{dfn-semiring-tensor}
Let $A,B,C$ be semirings and let ${_A}M_B$, ${_B}N_C$ be bimodules.
A tensor product $M\otimes_B N$ is an $(A,C)$-bimodule equipped with a map $(m,n)\mapsto m\otimes n$ universal among maps to $(A,C)$-bimodules that preserve addition and zero in each variable, are balanced over $B$, and respect the outer $A$- and $C$-actions.
\end{dfn}

\begin{prop}[{\cite[theorem~2.3]{katsov-nam-morita}}]\label{prop-semiring-tensor}
These tensor products exist and satisfy the natural adjunction
\[
\Hom_{A,C}(M\otimes_B N,P)
\cong \Hom_{A,B}(M,\Hom_C(N,P)),
\]
where $P$ is an $(A,C)$-bimodule and the inner Hom uses right $C$-module homomorphisms.
\end{prop}

\begin{proof}
To construct the tensor product, start with the free commutative monoid on $M\times N$ and take the quotient by the congruence generated by
\begin{align*}
(m+m',n)&\sim(m,n)+(m',n),\\
(m,n+n')&\sim(m,n)+(m,n'),\\
(0,n)&\sim0,\qquad(m,0)\sim0,\\
(mb,n)&\sim(m,bn).
\end{align*}
Define outer actions on generators by
\[ a(m\otimes n)=(am)\otimes n,\qquad (m\otimes n)c=m\otimes(nc), \]
and extend additively. These actions respect the generating relations because the actions on $M$ and $N$ commute with the $B$-actions.
The bimodule axioms follow from those on $M$ and $N$.
The free-monoid universal property and the defining relations give the required universal property.

It remains to establish the adjunction.
For $f$ in the Hom set on the left, define $\bar f$ by $\bar f(m)(n)=f(m\otimes n)$.
The outer actions on the tensor product ensure that $\bar f(m)$ is right $C$-linear and that $\bar f$ is left $A$-linear.
The balancing relation ensures that $\bar f$ is also right $B$-linear, where the action on the homomorphism module is $(g b)(n)=g(bn)$.
Conversely, given $g:M\to\Hom_C(N,P)$, the correspondence $(m,n)\mapsto g(m)(n)$ is balanced and therefore induces a unique $(A,C)$-linear map on the tensor product.
It is immediate from the formulas that these constructions are inverse and commute with precomposition and postcomposition.
Consequently they give the required natural isomorphism.
\end{proof}

The same construction gives ${_A}M_B\otimes {_B}N$ as a left $A$-module and $M_A\otimes {_A}N_B$ as a right $B$-module.

The next proposition gives the associativity and unit isomorphisms which we shall need to define the Morita skeleton.
Associativity is proved in~\cite[proposition~2.4]{katsov-nam-morita}, and we include the short argument together with the verification of the unit isomorphisms.

\begin{prop}\label{prop-tensor-associativity}
For semirings $A$, $B$, $C$, $D$ and semiring bimodules ${_A}M_B$, ${_B}N_C$, and ${_C}P_D$ there are canonical bimodule isomorphisms
\[
(M\otimes_B N)\otimes_C P\cong M\otimes_B(N\otimes_C P),
\quad A\otimes_A M\cong M\cong M\otimes_B B.
\]
\end{prop}

\begin{proof}
For associativity, observe that both bracketings have the same universal property: they represent maps on triples which are additive and preserve zero in each variable, are balanced over $B$ and $C$, and respect the outer actions.
Consequently the correspondence $(m\otimes n)\otimes p\mapsto m\otimes(n\otimes p)$ induces an isomorphism, whose inverse is given by the reverse correspondence.
For the unit isomorphisms, the maps $a\otimes m\mapsto am$ and $m\otimes b\mapsto mb$ are well-defined by the balancing relations.
It is immediate that they have respective inverses $m\mapsto1_A\otimes m$ and $m\mapsto m\otimes1_B$.
\end{proof}

Note that these constructions have been carried out for semiring modules.
Because the asterate axioms for Kleene modules are quasi-identities, they need not be preserved by quotients, so the construction above does not establish closure of Kleene modules under tensor products.
For the purposes of Morita equivalence we shall verify the asterate axioms for the particular bimodules which witness the claimed equivalences.

\subsection{Failure of an adjunction for Kleene modules}

The example of a homomorphism module which fails to satisfy the asterate axioms for Kleene modules raises the question of whether some other construction could give a right adjoint within the categories of Kleene modules.
The following proposition shows that in general this is not possible, even though the corresponding adjunction for semiring modules always exists by proposition~\ref{prop-semiring-tensor}.
Write $\RMod{R}$ for the category of right Kleene $R$-modules and module homomorphisms, and $\mathsf{JSL}_0$ for join-semilattices with zero and homomorphisms preserving zero and finite joins.
The latter category is canonically isomorphic to the category of right Kleene $\mathbb B$-modules (via forgetting the scalar action and via restoring it by defining $m\cdot0=0$ and $m\cdot1=m$); the asterate axiom for Kleene modules is automatic for this action.

\begin{prop}[A Kleene tensor functor with no right adjoint]
\label{prop-no-kleene-right-adjoint}
Let $R=\operatorname{Reg}(\{a\}^*)$ be the Kleene algebra of regular languages with alphabet $\{a\}$ and let $N={_R}R_{\mathbb B}$, with its regular left action and Boolean right action.
The ordinary semiring tensor product defines a functor
\[
-\otimes_R N:\RMod{R}\longrightarrow\mathsf{JSL}_0
\]
which has no right adjoint.
\end{prop}

\begin{proof}
The bimodule $N$ is Kleene on both sides.
The natural isomorphism $M\otimes_R N\cong U(M)$, given by $m\otimes L\mapsto mL$ with inverse $m\mapsto m\otimes1$, identifies this tensor functor with the forgetful functor
\[ U:\RMod{R}\longrightarrow\mathsf{JSL}_0. \]
It therefore suffices to show that $U$ has no right adjoint, and we shall do this by exhibiting a coequalizer which it does not preserve.
Here $a$ abbreviates the singleton language $\{a\}$, and $1=\{\epsilon\}$.
Consider the right $R$-module homomorphisms $\id_R$ and $\mu_a$, where the latter is defined by $\mu_a(L) = aL$ for $L \in R$.
Define the nonemptiness map $q:R\to\mathbb B$ by $q(\varnothing)=0$ and $q(L)=1$ for nonempty $L$.
This is clearly a semiring homomorphism.
Equip $\mathbb B$ with the right $R$-action $b\cdot L=bq(L)$.
This is a Kleene module: every $L^*$ is nonempty, so $b\cdot L^*=b$ for every $b\in\mathbb B$.
The map $q$ is right $R$-linear and satisfies $q\mu_a=q$.

We claim that $q$ is the coequalizer of $\mathrm{id}_R$ and $\mu_a$ in $\RMod{R}$.
Let $f:R\to M$ be a morphism with $f\mu_a=f$, and put $m=f(1)$.
Then $f(L)=mL$ and $ma=m$.
The asterate axiom for Kleene modules gives $ma^*\le m$.
Given a nonempty language $L$, choose $n\ge0$ with $a^n\in L$.
Since $a^n\le L\le a^*$, we have
\[ m=ma^n\le mL\le ma^*\le m. \]
Consequently $f(L)=m$ for all nonempty $L$.
There is a unique factorization $f=\bar f q$, where $\bar f(0)=0$ and $\bar f(1)=m$.
The map $\bar f$ preserves zero and joins because $M$ is additively idempotent.
It is right $R$-linear because $mL=m$ for nonempty $L$ and $m\varnothing=0$.
This proves the claim.

In $\mathsf{JSL}_0$, however, the map
\[
h:R\longrightarrow\mathbb B,\qquad
h(L)=\begin{cases}0,&L\text{ is finite},\\1,&L\text{ is infinite},\end{cases}
\]
preserves zero and finite joins and satisfies $h\mu_a=h$, since prefixing by $a$ preserves finiteness.
It does not factor through $U(q)$: we have $q(1)=q(a^*)=1$, but $h(1)=0$ and $h(a^*)=1$.
Thus $U(q)$ is not a coequalizer of $U(\mathrm{id}_R)$ and $U(\mu_a)$.
If $U$ had a right adjoint, $U$ would be a left adjoint and would preserve this coequalizer~\cite[V.5]{maclane-working}, a contradiction.
\end{proof}

Thus the difficulty in this example is not that the tensor functor fails to take Kleene modules to Kleene modules, but that it has no right adjoint in the indicated categories.
There is no conflict with proposition~\ref{prop-semiring-tensor}, because that proposition deals with the larger categories of semiring modules.
In the Morita arguments which follow we continue to use semiring tensor products, checking the Kleene module requirement for the bimodules which occur.

\section{Morita Equivalence of Semirings and Kleene Algebras}

\begin{dfn}\label{dfn-morita}
The \emph{Morita skeleton of semirings} has semirings as objects and isomorphism classes of $(A,B)$-bimodules as morphisms from $A$ to $B$.
Its identity at $A$ is the regular bimodule ${_A}A_A$, and composition is ordinary semiring tensor product.
Semirings $A,B$ are \emph{Morita-equivalent} when there are bimodules ${_A}M_B$, ${_B}N_A$ with
\[ M\otimes_B N\cong A,\qquad N\otimes_A M\cong B. \]
Kleene algebras are \emph{Kleene Morita-equivalent} when such witnesses can be chosen to be Kleene bimodules.
\end{dfn}

We use the term \emph{Kleene Morita groupoid} for the subcategory with just the invertible semiring bimodule classes between Kleene algebras.
To see that this does indeed give a groupoid, we use the following observation.
By standard semiring Morita theory, every invertible bimodule is finitely generated and projective on both sides~\cite{katsov-nam-morita}.
Every finitely generated projective module over the reduct of a Kleene algebra is a Kleene module.
Indeed, such a module is a retract of a finite free module and hence is isomorphic to a submodule of it, and the asterate axioms for Kleene modules hold coordinatewise in the free module and are inherited by submodules.
Consequently identities, inverses, and composites of invertible bimodules have the required Kleene asterate properties.

In the case of rings one obtains a bicategory structure by taking $2$-cells to be homomorphisms between modules~\cite{brouwer-bicategorical-morita}.
The same bimodule construction applies to semirings; we use only isomorphism classes and do not develop the bicategorical structure here.

Semiring homomorphisms determine morphisms in the semiring Morita skeleton via the following notion of homomorphism module,
which realizes homomorphisms as particular bimodules.

\begin{dfn}
Let $A$, $B$ be semirings (Kleene algebras) and $h \mathrel{:} A \rightarrow B$ be a homomorphism.
The \emph{homomorphism module} $E_h$ is a semiring (Kleene) $(A, B)$-bimodule whose elements and addition are given by the additive reduct of $B$, with the following multiplication actions on $b \in B$:
For $a \in A$, $a \cdot_{E_h} b = h(a) \cdot_B b$; and for $b' \in B$, $b \cdot_{E_h} b' = b \cdot_B b'$. 
\end{dfn}

The next proposition verifies that $E_h$ has the structure of a semiring (Kleene) bimodule.

\begin{prop}
Let $A$ and $B$ be semirings (Kleene algebras) and $h : A \rightarrow B$ be a homomorphism.
Then $E_h$ satisfies the axioms of a semiring (Kleene) $(A, B)$-bimodule.
\end{prop}

\begin{proof}
The additive monoid of $E_h$ is that of $B$, and is idempotent when $B$ is additively idempotent.
The axioms for the right action of $B$ are immediate, and the axioms for the left action of $A$ follow because $h$ is a homomorphism.
The fact that the two actions commute is an instance of associativity in $B$.
In case $A$ and $B$ are Kleene algebras, the asterate inequality for the right action follows from the fact that $B$ is a Kleene algebra, while the asterate inequality on the left is verified by the following computation,
where $a \in A$ and $b \in B$.
If $a \cdot_{E_h} b = h(a)b \le b$, then $h(a)^* b \le b$ because $B$ is a Kleene algebra, and hence $a^* \cdot_{E_h} b = h(a^*)b \le b$ because $h(a)^* = h(a^*)$ since $h$ is a homomorphism of Kleene algebras.
\end{proof}

Next we verify that the tensor product of homomorphism bimodules cor\-res\-ponds in the semiring Morita skeleton to the homomorphism module of the composition,
so that homomorphism modules define a functor from semirings and homomorphisms to the semiring Morita skeleton.
This functor need not be faithful, because homomorphisms which are the same up to conjugation by an invertible element can give isomorphic bimodules.
A homomorphism of Kleene algebras yields a Kleene bimodule, but that bimodule need not be invertible.

\begin{prop}
Let $A$, $B$, $C$ be semirings (Kleene algebras) and $f \mathrel{:} A \rightarrow B$ and $g \mathrel{:} B \rightarrow C$ be homomorphisms.
Then $E_{g \circ f} \cong E_f \otimes E_g$.
\end{prop}

\begin{proof}
The balanced map $(b,c)\mapsto g(b)c$ induces a bimodule homomorphism
\[ E_f\otimes_B E_g\longrightarrow E_{g\circ f}. \]
Its inverse is $c\mapsto1_B\otimes c$.
Balancing gives $b\otimes c=1_B\otimes g(b)c$, proving the inverse identities.
Left $A$-linearity follows from $g(f(a)b)c=g(f(a))g(b)c$; right $C$-linearity and additivity are immediate.
\end{proof}

Semiring Morita equivalence is equivalent to equivalence of the corresponding semiring module categories~\cite[theorem~4.12]{katsov-nam-morita}.
For witnesses ${}_K M_S$, ${}_S N_K$, the functors on left modules are
\[ P\longmapsto N\otimes_K P,\qquad Q\longmapsto M\otimes_S Q. \]
To verify that these are inverse up to natural isomorphism, one uses associativity of the tensor product together with the isomorphisms witnessing Morita equivalence.
In the preceding equivalences, the module categories are categories of semiring modules.

\begin{prop}\label{prop-matrix-morita}
Let $K$ be a Kleene algebra and $n\ge1$. Equip $A=M_n(K)$ with its standard matrix asterate~\cite{kozen-completeness}. Then $K$ and $A$ are Kleene Morita-equivalent.
\end{prop}

\begin{proof}
We shall use the column $(A,K)$-bimodule $C=K^{n\times1}$ and the row $(K,A)$-bimodule $R=K^{1\times n}$ to witness the equivalence.
Multiplication gives bimodule isomorphisms
\[ C\otimes_K R\cong A,\qquad R\otimes_A C\cong K. \]
For the first, an inverse sends a matrix $X$ to $\sum_{i,j} e_iX_{ij}\otimes e_j^{\mathsf T}$, where $e_i$ are the standard columns.
For the second, an inverse sends $k$ to $ke_1^{\mathsf T}\otimes e_1$.
Balancing with matrix units shows
\[ e_i^{\mathsf T}\otimes e_j
=e_1^{\mathsf T}E_{1i}\otimes e_j
=e_1^{\mathsf T}\otimes E_{1i}e_j
=\begin{cases}e_1^{\mathsf T}\otimes e_1,&i=j,\\0,&i\ne j.\end{cases} \]
Balancing also with $kE_{11}$ moves scalar coefficients across this last tensor, verifying the inverse identities on arbitrary rows and columns.

It remains to check that these are Kleene bimodules.
For the scalar $K$-actions, the asterate Kleene module axioms hold coordinatewise.
For the left $A$-action on $C$, embed a column as the first column of a matrix with all other columns zero; $Av\le v$ then implies $A^*v\le v$ using the fact that matrices form a Kleene algebra~\cite{kozen-completeness}. 
Embedding a row as the first row proves the right $A$-action case.
Thus both witnesses are Kleene bimodules.
\end{proof}

There is a well-known characterization of Morita equivalence for rings which extends to semirings~\cite{lam-lectures,katsov-nam-morita}.
To state it we need the following definition:
\begin{dfn}
Let $R$ be a semiring and $e \in R$.
The element $e$ is \emph{multiplicatively idempotent} precisely when $e^2 = e$, and \emph{full} if and only if $ReR = R$, where $ReR$ denotes the ideal of finite sums $\sum_i r_i e s_i$.
\end{dfn}

\begin{prop}[{\cite[proposition~5.1 of the preprint]{katsov-nam-radicals}}] \label{semiring-morita-full-idempotent}
Semirings $R,S$ are Morita-equivalent if and only if there exist a positive integer $n$ and a full idempotent $e$ in the semiring $M$ of $n \times n$ matrices over $R$ such that $S \cong eMe$.
Here $eMe$ has multiplicative identity $e$.
\end{prop}

In order to apply this characterization to Kleene algebras, we need to know that an idempotent corner carries an appropriate asterate.
There are connections here with earlier work on typed Kleene algebras.
Kozen~\cite[section~2.2]{kozen-typed} develops $1$-free Kleene algebras in terms of positive iteration $a^+=aa^*$.
M\"oller~\cite[section~3.1]{moeller-typed} studies restriction and corestriction by multiplication with types, which in his complete-lattice setting are idempotent subidentities.
For our purposes the idempotent must be allowed to be arbitrary, with no assumption that $e\le1$, and the Kleene algebra need not be complete.

\begin{lem}[Kleene algebra structure on an idempotent corner] \label{KA-idempotent-corner}
Let $A$ be a Kleene algebra and let $e \in A$ satisfy $e^2=e$.
The semiring $B=eAe$, with unit $e$, becomes a Kleene algebra when equipped with the asterate
\[ a^\dagger=ea^*e \qquad (a\in B), \]
where $*$ is the asterate of $A$.
\end{lem}

\begin{proof}
The corner $B$ is an additively idempotent semiring, and its natural order is the restriction of the natural order of $A$.
For $a\in B$ we have $ea=a=ae$, and hence
\begin{align*}
e+aa^\dagger &= e+aea^*e=e(1+aa^*)e=ea^*e=a^\dagger,\\
e+a^\dagger a &= e+ea^*ea=e(1+a^*a)e=ea^*e=a^\dagger.
\end{align*}
Let $a,b,c\in B$. If $b+ac\le c$, the asterate axioms in $A$ yield $a^*b\le c$.
Multiplying on the left by $e$ gives 
\[ a^\dagger b=ea^*eb=ea^*b\le ec=c. \]
Similarly, if $b+ca\le c$, then $ba^*\le c$ in $A$, so
\[ ba^\dagger=bea^*e=ba^*e\le ce=c. \]
These are the Kleene algebra axioms with multiplicative identity $e$.
\end{proof}

For $a\in eAe$, the corner asterate also satisfies
\[ a^\dagger=ea^*e=e+eaa^*e=e+aa^*=e+a^+, \]
using $ea=a=ae$ and $aa^*=a^*a$.
In other words, passing to the corner retains positive iteration while replacing the multiplicative identity by $e$, which explains the connection with the $1$-free Kleene algebras of~\cite{kozen-typed}.

We now turn to the rigidity theorem.
Our goal is to show that the bimodules witnessing a semiring Morita equivalence already satisfy the asterate axioms for Kleene modules.
For this purpose we first need an observation about how equations in a module can be recovered from equations in its tensor products.
When the module is additively idempotent, the same observation will allow us to recover its natural order.

\begin{lem}[Tensor separation]\label{lem-tensor-separation}
Let $R,S$ be semirings and let ${}_R M_S$, ${}_S N_R$ be bimodules with bimodule isomorphisms
\[ \alpha:M\otimes_S N\longrightarrow R,\qquad
   \beta:N\otimes_R M\longrightarrow S. \]
For $x,y\in M$, if $x\otimes n+y\otimes n=y\otimes n$ for every $n\in N$, then $x+y=y$.
The same conclusion holds if $n\otimes x+n\otimes y=n\otimes y$ for every $n\in N$.
Both assertions also hold with equalities in place of the displayed additive relations.
\end{lem}

\begin{proof}
We write $\alpha(m,n)=\alpha(m\otimes n)$ and $\beta(n,m)=\beta(n\otimes m)$.
For fixed $p\in M$, the map
\[ m\otimes n\longmapsto m\beta(n,p) \]
is well-defined and additive: balancing follows from the left $S$-linearity of $\beta$.
The first hypothesis therefore implies
\[ x\beta(n,p)+y\beta(n,p)=y\beta(n,p). \]
Every element of $N\otimes_R M$ is a finite sum of pure tensors, so additivity gives the same equation with $\beta(z)$ in place of $\beta(n,p)$ for any $z\in N\otimes_R M$.
Taking $z=\beta^{-1}(1_S)$ gives $x+y=y$.
For the second assertion we proceed symmetrically, using the additive maps
\[ n\otimes m\longmapsto\alpha(p,n)m \]
for $p\in M$ and then evaluating at $\alpha^{-1}(1_R)$.
The same computations prove the assertions with equalities in place of additive relations.
\end{proof}

\begin{thm}[Morita rigidity of Kleene algebras over semirings]\label{KA-strong-morita}
Let $K$ and $S$ be Kleene algebras.
Every pair of semiring bimodules witnessing Morita equivalence of their semiring reducts is already a pair of Kleene bimodules witnessing the same equivalence.
In particular, $K$ and $S$ are Morita-equivalent via Kleene bimodules whenever their semiring reducts are Morita-equivalent.
\end{thm}

\begin{proof}
Fix bimodules ${}_K M_S$, ${}_S N_K$ witnessing the semiring Morita equivalence, with tensor isomorphisms $\alpha,\beta$ as in lemma~\ref{lem-tensor-separation}.
We first verify the asterate axiom for the left action of $K$ on $M$.
Let $a\in K$ and $m\in M$, and suppose that $am\le m$.
For every $n\in N$, left linearity gives
\[ a\alpha(m,n)\le\alpha(m,n). \]
The asterate axioms in $K$ imply $a^*\alpha(m,n)\le\alpha(m,n)$, and hence
\[ \alpha(a^*m,n)\le\alpha(m,n). \]
Since $\alpha$ is an additive isomorphism, it reflects the natural order, so $a^*m\otimes n\le m\otimes n$ for every $n$.
Lemma~\ref{lem-tensor-separation} gives $a^*m\le m$.
For the right $S$-action on $M$, apply the same argument to $\beta(n,m)$ and use the second separation assertion.
Interchanging $M,N$ and $K,S$ proves the asterate axioms for $N$.
Thus the original bimodules and tensor isomorphisms witness Kleene Morita equivalence.
\end{proof}

The proof of the asterate axiom for multiplication on the left for $M$ uses only the asterate for $K$, and the same is true for the symmetric proof of the asterate axiom for multiplication on the right in the case of $N$.
Consequently, even when $S$ is initially only a semiring, ${}_K M$ and $N_K$ are Kleene modules.
We shall use these facts to construct an asterate on $S$ directly.
To carry out the construction we shall need the tensor pairings to satisfy two identities giving compatibility with the module actions.
These identities need not hold for the isomorphisms $\alpha,\beta$ with which we started, but the following lemma shows that we can arrange for them by changing one of the isomorphisms.

\begin{lem}[Compatible pairings]\label{lem-compatible-pairings}
For the semiring Morita witnesses of lemma~\ref{lem-tensor-separation}, one may replace $\beta$ by a bimodule isomorphism $\gamma:N\otimes_R M\longrightarrow S$ such that
\[ \alpha(m,n)p=m\gamma(n,p),\qquad
   \gamma(n,m)q=n\alpha(m,q) \]
for $m,p\in M$ and $n,q\in N$.
The modules and their scalar actions are unchanged.
\end{lem}

\begin{proof}
Choose finite tensor presentations
\[ \beta^{-1}(1_S)=\sum_i n_i\otimes m_i,\qquad
   \alpha^{-1}(1_R)=\sum_j p_j\otimes q_j, \]
and define additive maps $\theta,\psi:M\longrightarrow M$ by
\[ \theta(m)=\sum_i\alpha(m,n_i)m_i,\qquad
   \psi(m)=\sum_j p_j\beta(q_j,m). \]
The balancing relations ensure that these definitions do not depend on the chosen tensor presentations.
We shall first verify that $\theta$ and $\psi$ are inverse maps, and then that $\theta$ preserves the module actions.
For the first of these verifications, note that left $S$-linearity and injectivity of $\beta$ give
\[ \beta(n,p)\beta^{-1}(1_S)=n\otimes p. \]
Applying the additive map $q\otimes x\mapsto\alpha(m,q)x$ yields
\[ \theta(m\beta(n,p))=\alpha(m,n)p. \]
Similarly, right $R$-linearity and injectivity of $\alpha$ give
\[ \psi(\alpha(m,n)p)=m\beta(n,p). \]
Summing these identities over the chosen presentations shows
\[ \theta\psi(m)=\sum_j\alpha(p_j,q_j)m=m,
   \qquad \psi\theta(m)=\sum_i m\beta(n_i,m_i)=m. \]
Thus $\theta$ and $\psi$ are inverse additive maps.

It remains to verify that $\theta$ is a bimodule homomorphism.
The fact that it is left $R$-linear follows immediately from its definition.
To see that it is right $S$-linear, let $s\in S$ and note that injectivity and bimodule linearity of $\beta$ imply
\[ s\beta^{-1}(1_S)=\beta^{-1}(1_S)s. \]
Applying $q\otimes x\mapsto\alpha(m,q)x$ gives $\theta(ms)=\theta(m)s$.
Hence $\theta$ is a bimodule automorphism.
Set
\[ \gamma(n,p)=\beta(n,\theta(p)). \]
The map $\gamma$ is a bimodule isomorphism because it is the composite of $\beta$ with $1_N\otimes\theta$.
Applying $\theta$ and using the preceding identities proves
\[ m\gamma(n,p)=\alpha(m,n)p. \]
To verify the second identity we shall use tensor separation.
For every $p\in M$ we compute that
\begin{align*}
 \alpha(p,\gamma(n,m)q)
 &=\alpha(p\gamma(n,m),q)\\
 &=\alpha(\alpha(p,n)m,q)
 =\alpha(p,n)\alpha(m,q)
 =\alpha(p,n\alpha(m,q)).
\end{align*}
Injectivity of $\alpha$ and the equality form of tensor separation, with the roles of the witnesses interchanged, give $\gamma(n,m)q=n\alpha(m,q)$.
\end{proof}

We next isolate the properties which a proposed value of $a^*$ must satisfy.
This will allow us to construct asterate one element at a time, and then to show that the existence of such values is preserved when we take finite sums.

\begin{lem}\label{lem-local-closure}
Let $T$ be an additively idempotent semiring.
Call $p\in T$ a \emph{Kleene closure} of $a\in T$ if
\[ 1\le p,\qquad ap\le p, \]
and, for every $b\in T$,
\[ ab\le b\ \Longrightarrow\ pb\le b,
   \qquad ba\le b\ \Longrightarrow\ bp\le b. \]
Such a closure is unique when it exists, and satisfies $pa\le p$.
Moreover the following hold.
\begin{itemize}
\item The element $0$ has closure $1$.
\item If $p$ is a closure of $a$ and $q$ is a closure of $bp$, then $pq$ is a closure of $a+b$.
\item If every element has a Kleene closure, these closures define a Kleene algebra expansion of $T$.
\end{itemize}
\end{lem}

\begin{proof}
Since $ap\le p$, we have $a(ap)\le ap$.
The defining implication for $p$ gives $p(ap)\le ap$, whence
\[ pa\le pap\le ap\le p. \]
If $q$ is another closure of $a$, then $p\le pq\le q$; symmetry gives equality.
The assertion about $0$ is immediate.

Suppose now that $p$ is a closure of $a$ and $q$ a closure of $bp$.
We have $1\le pq$, and
\[ a(pq)\le pq,\qquad b(pq)=(bp)q\le q\le pq. \]
If $(a+b)c\le c$, then $pc\le c$ and $(bp)c\le bc\le c$.
It follows that $qc\le c$, so $(pq)c\le pc\le c$.
If $c(a+b)\le c$, then $cp\le c$ and $c(bp)\le cp\le c$.
Thus $cq\le c$ and $c(pq)\le cq\le c$.
This proves the sum assertion.
Finally, the closure conditions, together with $pa\le p$, are the two implicational asterate axioms and the unfold inequalities for Kleene algebra.
They imply the usual unfold equalities as well: $1+ap\le p$, while $a(1+ap)\le1+ap$ gives $p\le p(1+ap)\le1+ap$.
The right unfold equality follows symmetrically.
\end{proof}

\begin{thm}[Strong Morita rigidity of Kleene algebras over semirings]\label{KA-morita-expansion}
Let $K$ be a Kleene algebra and $S$ a semiring Morita-equivalent to the semiring reduct of $K$.
There is a unique expansion $S'$ of $S$ to a Kleene algebra.
Moreover, every pair of semiring bimodules witnessing the given Morita equivalence is already a pair of Kleene bimodules for $K$ and $S'$.
Thus the expansion changes neither the semiring operations nor the witnessing modules, actions, or tensor isomorphisms.
\end{thm}

\begin{proof}
Fix bimodules ${}_K M_S$, ${}_S N_K$ witnessing the given semiring Morita equivalence, with tensor isomorphisms $\alpha,\beta$.
We first check that $S$ is additively idempotent.
The $K$-actions make $M$ and $N$ additively idempotent.
Every element of $N\otimes_K M$ is a finite sum of pure tensors, and pure tensors are additively idempotent, so this tensor product is also additively idempotent.
Since $\beta$ is an additive isomorphism onto $S$, it follows that $S$ is additively idempotent as well.
As observed after theorem~\ref{KA-strong-morita}, ${}_K M$ and $N_K$ satisfy the asterate axioms for Kleene modules.
Choose the compatible pairing $\gamma$ given by lemma~\ref{lem-compatible-pairings}.

We first construct a Kleene closure for every element of the form $s=\gamma(n,m)$.
Put
\[ k=\alpha(m,n),\qquad p=1_S+\gamma(n,k^*m). \]
Compatibility and balancing give
\[ \gamma(n,m)\gamma(q,x)=\gamma(n,\alpha(m,q)x). \]
Therefore
\[ sp=\gamma(n,m)+\gamma(n,kk^*m)
   \le\gamma(n,k^*m)\le p, \]
where we used $1_K\le k^*$ and $kk^*\le k^*$.

It remains to verify the two asterate-style implications in the definition of a Kleene closure.
For the first, let $b\in S$ and suppose that $sb\le b$.
Multiplying the inequality on the left by $m$ and using compatibility gives
\[ k(mb)\le mb. \]
Hence because $M$ is a left Kleene module we have that $k^*(mb)\le mb$.
Consequently
\[ \gamma(n,k^*m)b
   =\gamma(n,k^*(mb))\le\gamma(n,mb)=sb\le b, \]
so $pb\le b$.
Similarly, if $bs\le b$, multiply on the right by $n$ to obtain $(bn)k\le bn$.
The asterate axiom in the right Kleene $K$-module $N$ gives $(bn)k^*\le bn$, and hence
\[ b\gamma(n,k^*m)
   =\gamma((bn)k^*,m)\le\gamma(bn,m)=bs\le b. \]
Thus $bp\le b$, proving that $p$ is a Kleene closure of $s$.

Having constructed closures for elements represented by pure tensors, we now extend the construction to all of $S$.
Since $\gamma$ is surjective, every element of $S$ is a finite sum of elements of the form $\gamma(n,m)$, and we proceed by induction on the number of summands.
The empty sum has closure $1_S$.
If $a$ has closure $p$ and $b=\gamma(n,m)$ is the next summand, then
\[ bp=\gamma(n,mp) \]
is again represented by a pure tensor and therefore has a closure $q$.
By lemma~\ref{lem-local-closure}, $pq$ is then a closure of $a+b$.
This completes the induction.

Define $s^{*_{S'}}$ to be the unique Kleene closure of $s$.
Lemma~\ref{lem-local-closure} proves that this is a Kleene algebra expansion, using exactly the original semiring operations of $S$.
Any Kleene asterate must give these same closures, so the expansion is unique and independent of all choices made in its construction.
Finally, theorem~\ref{KA-strong-morita} applies to $K$ and $S'$ and shows that the original witnesses, with the original tensor isomorphisms $\alpha,\beta$, are Kleene Morita witnesses.
Note that although we used $\gamma$ to construct the asterate, the original isomorphisms $\alpha,\beta$ still witness the equivalence.
\end{proof}

This proof does not use the matrix-corner characterization of semiring Morita equivalence.
That characterization, together with lemma~\ref{KA-idempotent-corner}, gives the alternative description
\[ s^{*_{S'}}=\phi^{-1}\bigl(e\,\phi(s)^{*_{M_n(K)}}e\bigr) \]
whenever $\phi:S\longrightarrow eM_n(K)e$ is a unital semiring isomorphism.
Uniqueness shows that the two constructions agree.

\subsection{Equivalence of categories of Kleene modules}
\label{sec-kleene-module-equivalence}

We have seen that Morita equivalence of semirings is characterized by equivalence of their categories of semiring modules.
It is natural to ask whether the analogous characterization holds for Kleene algebras when we consider only Kleene modules as objects, and we establish that it does.
Write $\LMod{K}$ for the category of left Kleene $K$-modules and module homomorphisms.
All projectivity assertions about semiring modules below refer to lifting along surjective homomorphisms.

Recall that a regular epimorphism is a morphism which is a coequalizer of some pair of morphisms.
An object is projective with respect to regular epimorphisms if every map from it to the codomain of a regular epimorphism lifts to the domain.
An object $P$ is finitely presentable if $\Hom(P,-)$ preserves filtered colimits.
These properties are expressed entirely in terms of the category, and hence are preserved by equivalences.
For background on finite presentability and the categorical treatment of quasivarieties, see~\cite[chapters~1 and~3]{adamek-rosicky-locally-presentable}.

\begin{lem}\label{lem-kleene-categorical-projectives}
Let $K$ be a Kleene algebra.
The finitely generated projective modules over its semiring reduct are precisely the objects of $\LMod{K}$ which are finitely presentable and projective with respect to regular epimorphisms.
\end{lem}

\begin{proof}
We first verify that the regular epimorphisms in $\LMod{K}$ are exactly the surjective homomorphisms.
Products and submodules of Kleene modules satisfy the asterate axioms, so the kernel pair of a homomorphism is again a Kleene module.
If a homomorphism is surjective, the usual factorization through its kernel pair shows that it is the coequalizer of that pair in $\LMod{K}$.
Conversely, every homomorphism factors through its image, which is a Kleene submodule of its codomain.
A regular epimorphism cannot factor through a proper subobject: if $q=ip$ coequalizes a pair and $i$ is a monomorphism, then $p$ coequalizes the same pair, and the universal property of $q$ makes $i$ an isomorphism.
Applying this observation to the image factorization shows that every regular epimorphism is surjective.

For any set $I$, the free semiring module $K^{(I)}$ of finitely supported functions from $I$ to $K$ satisfies the asterate axioms coordinatewise.
It therefore remains free on $I$ in $\LMod{K}$.
Filtered colimits in this category are also computed as for semiring modules.
To see this, suppose that $am+m=m$ in a filtered colimit.
Choose a representative for $m$ at one stage of the diagram; the displayed equality holds at some further stage, where the asterate axiom gives $a^*m+m=m$.
It follows that this equality holds in the colimit as well.
Thus finite free modules are finitely presentable, since maps from $K^n$ are specified by $n$ elements, and they are projective with respect to regular epimorphisms, since these are surjective.
Both properties pass to retracts.
Every finitely generated projective semiring module is a retract of a finite free module, so it has the asserted categorical properties.

For the converse, let $P$ have these properties and consider the canonical surjection
\[ \pi:K^{(|P|)}\longrightarrow P, \]
where $|P|$ denotes the underlying set of $P$ and the basis element indexed by $p$ is sent to $p$.
Projectivity gives a section $j:P\to K^{(|P|)}$.
The codomain is the filtered union of the finite free submodules $K^J$, for finite subsets $J\subseteq |P|$.
Finite presentability therefore implies that $j$ factors through some $K^J$.
Restricting $\pi$ to this finite free module exhibits $P$ as a retract of $K^J$, as required.
\end{proof}

To show that the tensor functors preserve the Kleene asterate axioms, we shall use a module version of the asterate axioms for matrices.
The standard matrix asterate is defined by the usual block formulas (see~\cite{kozen-completeness}); the following proof verifies the additional assertion about its action on a module.

\begin{lem}\label{lem-matrix-module-induction}
Let $M$ be a left Kleene $K$-module and $n\ge1$.
Then $M^n$ is a left Kleene $M_n(K)$-module under the usual action on columns.
Moreover, if $e\in M_n(K)$ satisfies $e^2=e$, then $eM^n$ is a left Kleene module over the corner $eM_n(K)e$, with the asterate of lemma~\ref{KA-idempotent-corner}.
\end{lem}

\begin{proof}
We prove by induction on $n$ that $Ax\le x$ implies $A^*x\le x$.
For $n=1$ this is exactly the asterate axiom on $M$.
For the induction step write
\[ A=\begin{pmatrix}B&b\\c&d\end{pmatrix},
   \qquad x=\binom yz,\qquad T=B+bd^*c, \]
where $B$ is an $(n-1)\times(n-1)$ matrix and $d\in K$.
The block formula for matrix asterate gives
\[ A^*x=\binom uv,\qquad
   u=T^*(y+bd^*z),\qquad v=d^*(z+cu). \]
Suppose that $Ax\le x$, so that
\[ By+bz\le y,\qquad cy+dz\le z. \]
The asterate axiom in $M$ gives $d^*z\le z$.
Consequently $d^*cy\le d^*z\le z$, and hence
\[ Ty=By+bd^*cy\le By+bz\le y. \]
By the induction hypothesis, $T^*y\le y$.
Since $y+bd^*z\le y$, we obtain $u\le y$.
It follows that $z+cu\le z$, and therefore
\[ v=d^*(z+cu)\le d^*z\le z. \]
This proves $A^*x\le x$.

For the assertion about corners, let $a\in eM_n(K)e$ and $x\in eM^n$ satisfy $ax\le x$.
The result just proved gives $a^*x\le x$, with asterate computed in $M_n(K)$.
Since $ex=x$, the corner asterate satisfies
\[ a^\dagger x=ea^*ex=ea^*x\le ex=x. \]
Thus the corner action also satisfies the asterate axiom.
\end{proof}

\begin{thm}[Characterization by categories of Kleene modules]
\label{thm-kleene-module-equivalence}
For Kleene algebras $K,S$, the following are equivalent:
\begin{enumerate}
\item $K$ and $S$ are Kleene Morita-equivalent;
\item their semiring reducts are Morita-equivalent;
\item $\LMod{K}$ and $\LMod{S}$ are equivalent categories.
\end{enumerate}
\end{thm}

\begin{proof}
The equivalence of the first two conditions is theorem~\ref{KA-strong-morita}.
We first show that the third condition implies the second.
Let
\[ E:\LMod{K}\longrightarrow\LMod{S} \]
be an equivalence with quasi-inverse $D$, and put $P=E({}_K K)$ and $Q=D({}_S S)$.
By lemma~\ref{lem-kleene-categorical-projectives}, $P$ is a finitely generated projective semiring $S$-module and $Q$ is a finitely generated projective semiring $K$-module.
Choose maps exhibiting $Q$ as a retract of $K^n$.
Applying $E$, which preserves finite coproducts, exhibits $S\cong E(Q)$ as a retract of $P^n$.
Every semiring $S$-module is a quotient of a coproduct of copies of $S$, and hence of copies of $P$.
Thus $P$ is a generator for the category of all left semiring $S$-modules, not merely for its Kleene modules.

The equivalence also gives an isomorphism of endomorphism semirings
\[ \End[S](P)\cong\End[K]({}_K K)\cong K^{\mathrm{op}}. \]
Here endomorphisms preserve the indicated left actions.
To verify that the first isomorphism preserves addition, note that finite direct sums are biproducts in the Kleene module categories.
Addition of morphisms is therefore determined by the categorical formula
\[ f+g=\nabla\circ(f\oplus g)\circ\Delta, \]
where $\Delta$ and $\nabla$ are the diagonal and codiagonal.
An equivalence preserves this formula, as well as composition and identities.
The second isomorphism identifies $k\in K$ with right multiplication by $k$ on ${}_K K$.
Consequently $P$ is a semiring progenerator with endomorphism semiring $K^{\mathrm{op}}$.
The semiring Morita theorem~\cite[section~4]{katsov-nam-morita} now gives the second condition.

Conversely, suppose that the semiring reducts are Morita-equivalent, with inverse bimodules ${}_S N_K$ and ${}_K L_S$.
The corresponding equivalence of semiring module categories is given by
\[ F(M)=N\otimes_K M,\qquad G(X)=L\otimes_S X. \]
It remains to show that these functors take Kleene modules to Kleene modules.
By semiring Morita theory~\cite[section~4]{katsov-nam-morita}, $N_K$ is finitely generated projective and the left action gives a semiring isomorphism
\[ S\cong\End[K](N_K). \]
Choose a splitting $N_K\cong eK^n$ for an idempotent $e\in M_n(K)$, where $K^n$ is a right module of columns; we may take $n\ge1$.
This identifies
\[ \End[K](N_K)\cong eM_n(K)e,\qquad
   N\otimes_K M\cong eM^n. \]
For the second identification, tensor the splitting of $N_K$ and use $K^n\otimes_K M\cong M^n$.
Under these identifications the left $S$-action is the action of the corresponding corner matrices.

The semiring isomorphism $S\cong eM_n(K)e$ preserves asterate.
Indeed, in any Kleene algebra $a^*$ is the least solution of $1+ax\le x$, and a semiring isomorphism preserves the natural order and this least-solution property.
Lemma~\ref{lem-matrix-module-induction} therefore shows that $N\otimes_K M$ is a Kleene $S$-module whenever $M$ is a Kleene $K$-module.
The same argument applies to $G$.
The unit and counit of the semiring Morita equivalence remain module isomorphisms in the full subcategories of Kleene modules.
Thus the restrictions of $F$ and $G$ give the equivalence in the third condition.
\end{proof}

This theorem is consistent with the failure of a right adjoint in proposition~\ref{prop-no-kleene-right-adjoint}.
In the present situation the bimodule is invertible, and its finite projectivity allows us to verify the Kleene asterate axioms by passing to a matrix corner.
The inverse Morita functor then supplies both adjoints to the restricted equivalence.

\section{Projective Automata and Morita Equivalence}
\label{sec-projective-automata}

The correspondence between finite automata and transition matrices suggests that Morita equivalence might also be useful for comparing categories of automata.
Before making this precise, however, we must decide what the objects and morphisms of these categories should be.
We shall first explain a difficulty with taking labelled graphs and bisimulations, and then give a definition in terms of modules and transition endomorphisms for which Morita equivalence does yield an equivalence of automaton categories.

\subsection{Markings and labelled bisimulations}
\label{sec-markings-bisimulations}

A \emph{marked Kleene algebra} is a pair $(K,S)$, where $S\subseteq K$ is a finite generating set for the Kleene algebra operations, with $0,1$ available as nullary operations (they need not belong to $S$).
One natural category to associate with such an algebra has as its objects finite $S$-labelled transition systems with one distinguished initial state and a set of accepting states.
A morphism is a relation satisfying the usual label-preserving back-and-forth conditions and relating only states that agree on acceptance.
We impose no requirement on initial states, and do not require the relation to be total.
With diagonal relations as identities and relational composition as composition, these objects and morphisms form a category.

Unfortunately, the resulting category can depend on the marking even when the underlying algebra is unchanged.
To see this, consider the Boolean Kleene algebra $\mathbb B=\{0,1\}$, for which both $\varnothing$ and $\{1\}$ are generating sets.
For the empty marking, an automaton with $a$ accepting and $b$ nonaccepting states has exactly $2^{a^2+b^2}$ endomorphisms.
Since there is an initial state, $a+b\ge1$, so exactly two endomorphisms occur precisely for one-state automata.
There are exactly two such isomorphism classes, distinguished by acceptance.
For the marking $\{1\}$, there are at least four: the one-state automata obtained by independently choosing acceptance and the presence of a $1$-labelled loop.
Each has only the empty and diagonal endomorphisms, and no two are isomorphic.
An equivalence preserves endomorphism sets and isomorphism classes, so these two categories are not equivalent.

We can avoid dependence on a marking by allowing labels from all of $K$.
There is still a difficulty with interpreting these labelled graphs as matrices, because bisimulation distinguishes a $0$-labelled edge from no edge, and parallel edges labelled $a,b$ from a single edge labelled $a+b$.
The corresponding transition matrices identify these presentations, so using matrices requires a change in our definition of the automaton category.

\subsection{Projective automata}

For the remainder of this section we work with right modules over the semiring reduct, and projectivity is understood in this category.
A finitely generated projective module is equivalently a retract of a finitely generated free module $K^n$; we allow the zero module.
We write $\RMod{K_{\mathrm{red}}}$ for the category of right semiring modules over $K_{\mathrm{red}}$.

For a free state module $K^n$, a transition matrix $A\in M_n(K)$ acts on columns by left multiplication.
A matrix $X\in M_{m\times n}(K)$ defines a morphism to a transition matrix $B\in M_m(K)$ when
\[ XA=BX. \]
Such maps are called linear intertwiners, and will take the place of relational bisimulations in the category which we define.
A Morita equivalence need not carry finitely generated free modules to finitely generated free modules, but it preserves finitely generated projective modules~\cite[proposition~4.8 and lemma~4.10]{katsov-nam-morita}.
This is the reason for allowing projective state modules in the following definition.

\begin{dfn}\label{dfn-projective-automata}
A \emph{projective automaton} over $K$ is a triple $(P,T,F)$, where $P$ is a finitely generated projective right $K$-module, $T:P\to P$ is a module endomorphism, and $F\subseteq P$ is an accepting submodule.
A morphism $(P,T,F)\to(Q,U,G)$ is a module homomorphism $f:P\to Q$ satisfying
\[ fT=Uf,\qquad f(F)\subseteq G. \]
These objects and morphisms form a category $\mathsf{PAut}(K)$ under ordinary composition.
The category $\mathsf{PAut}_{\mathrm{ref}}(K)$ has the same objects and those morphisms that also reflect acceptance, meaning
\[ f^{-1}(G)=F. \]
\end{dfn}

To see that the second category is well-defined, note that the inverse-image condition is preserved by composition.
In either definition the accepting submodule $F$ may be arbitrary; it need not be finitely generated or projective, nor do we require it to be invariant under $T$.
For ordinary free state modules, a set $J$ of accepting coordinates determines the submodule
\[ F_J=\{v\in K^n:v_i=0\text{ for }i\notin J\}. \]
The reason for allowing more general accepting submodules is that an equivalence of module categories need not preserve this particular form of submodule.
Note also that both $P$ and $F$ are Kleene modules, because $K^n$ is a Kleene module, $P$ is isomorphic to a submodule of $K^n$, and the asterate axioms for Kleene modules are inherited by submodules.

One may add a distinguished initial element $p_0\in P$ to each object and impose no condition on it in the definition of a morphism.
This does not change the category up to equivalence, since the functor which forgets the initial element is fully faithful and has a section obtained by choosing $p_0=0$.
The identity on $P$ is an isomorphism between any two choices of initial element.
Here the initial element is allowed to be arbitrary, since a projective module need not have a basis from which to choose it.

\begin{thm}[Morita invariance of projective automata]\label{thm-projective-automata-morita}
Let $K$ and $L$ be Kleene algebras with Morita-equivalent semiring reducts. Then
\[
\mathsf{PAut}(K)\simeq\mathsf{PAut}(L),
\qquad
\mathsf{PAut}_{\mathrm{ref}}(K)\simeq
\mathsf{PAut}_{\mathrm{ref}}(L).
\]
The same conclusion holds when an initial element is included and ignored by morphisms.
\end{thm}

\begin{proof}
The Morita bimodules induce an equivalence
\[ E:\RMod{K_{\mathrm{red}}}\longrightarrow\RMod{L_{\mathrm{red}}} \]
with quasi-inverse $D$.
This equivalence preserves finitely generated projective modules~\cite[proposition~4.8 and lemma~4.10]{katsov-nam-morita}.
To define the accepting submodule after applying the equivalence, consider the inclusion $i:F\hookrightarrow P$.
The map $E(i)$ is a monomorphism and hence an injective module homomorphism, so we may identify its domain with its image $F_E\subseteq E(P)$.
We now send $(P,T,F)$ to $(E(P),E(T),F_E)$ and a morphism $f$ to $E(f)$.

We must check that this assignment preserves the requirements on morphisms.
The transition equation follows from functoriality, which gives $E(f)E(T)=E(U)E(f)$.
If $f(F)\subseteq G$, the restriction $f|_F:F\to G$ gives a commutative square with the two inclusions, and applying $E$ shows $E(f)(F_E)\subseteq G_E$.
Thus this construction defines a functor on $\mathsf{PAut}$.
For the category in which morphisms also reflect acceptance, observe that the square
\[
\begin{tikzcd}
F \arrow[r,"f|_F"] \arrow[d,hook] & G \arrow[d,hook] \\
P \arrow[r,"f"'] & Q
\end{tikzcd}
\]
is a pullback precisely when $f^{-1}(G)=F$.
Because equivalences preserve pullbacks, the same construction gives a functor on $\mathsf{PAut}_{\mathrm{ref}}$ as well.

It remains to verify that these functors are equivalences, with inverses up to natural isomorphism obtained by applying $D$.
Indeed, the natural isomorphisms $DE\cong\mathrm{Id}$ and $ED\cong\mathrm{Id}$ commute with transition endomorphisms by naturality, and carry the transported accepting subobjects onto the original accepting subobjects by naturality with respect to their inclusions.
Their components are therefore isomorphisms in both automaton categories.
The assertion concerning initial elements follows from the forgetful equivalences described above.
\end{proof}

Note that the proof of the theorem uses only semiring Morita theory.
The additional Kleene structure allows us to take asterates of transition endomorphisms, as can be seen by writing $P=eK^n$ for an idempotent $e\in M_n(K)$.
Then
\[ \End[K](P)\cong eM_n(K)e, \]
and lemma~\ref{KA-idempotent-corner} equips this endomorphism semiring with the asterate
\[ A\mapsto eA^{*_{M_n(K)}}e. \]
A morphism $eK^n\to dK^m$ is represented by a rectangular matrix $X$ satisfying $X=dXe$, with the additional transition equation $XA=BX$.
Thus the same idempotent corners which occur in Morita theory also give a matrix description of projective automata.
For related work on rectangular matrices and typed composition see Kozen's typed Kleene algebra~\cite{kozen-typed}, and for results on transferring equational reasoning between the typed and untyped settings see Pous~\cite{pous-untyping}.
These results concern the typed algebraic setting; the equivalence of automaton categories proved here follows from semiring Morita theory.

Another possibility is to specify acceptance by an output map $\lambda:P\to O$ into a fixed module $O$.
In this case the Morita equivalence must be applied to $O$ as well as to $P$, giving the output map $E(\lambda):E(P)\to E(O)$.
This matters even if $O=K$, since its image need not be $L_L$.
Using an accepting submodule avoids having to choose such an output module.
The theorem therefore concerns the projective automata defined above, and does not establish an equivalence of the original labelled-bisimulation categories.
Whether one can equip these projective automata with an appropriate language semantics which is also preserved by the equivalence is a separate question.

\section{Conclusion and further directions}

Theorem~\ref{KA-morita-expansion} shows that the existence of a Kleene algebra expansion is invariant under semiring Morita equivalence.
More specifically, every semiring which is Morita-equivalent to the reduct of a Kleene algebra has a unique Kleene algebra expansion, and every original pair of witnessing semiring bimodules is already a pair of Kleene bimodules for that expansion.
The proof constructs asterate directly from compatible tensor pairings and finite sums, without using the matrix-corner characterization.

The application to projective automata in theorem~\ref{thm-projective-automata-morita} suggests further questions about the relationship between Morita equivalence and computation.
In particular, it would be of interest to give these automata a language semantics and to understand how the linear intertwiners used here relate to the usual notions of simulation and bisimulation.

At the time of this writing Morita equivalence is a new theoretical tool being introduced into the study of Kleene algebras and related structures,
but it does have some intriguing potential for applications to the theory of programming language semantics.
One observation is that it is possible to freely change scalars from one Kleene algebra $K$ to a Morita-equivalent one $S$, because if bimodules ${_K}M_S$ and ${_S}N_K$ witness the Morita equivalence of $K$ and $S$,
then for any left module ${_K}P$ we can reversibly pass to the $S$-module ${_S}Q = {_S}N_K \otimes {_K}P$, the reversibility being witnessed by the fact that
\[ {_K}M_S \otimes {_S}Q \cong ({_K}M_S \otimes {_S}N_K) \otimes {_K}P \cong {_K}K_K \otimes {_K}P \cong {_K}P. \]
Changes of scalars for right modules and bimodules are handled similarly.

It is also of interest to investigate what additional structure can be added to a Kleene algebra (such as order-meets, to form a Kleene lattice, or order-residuals, to form an action algebra; see~\cite{kozen-action}) while preserving
Morita rigidity in the sense that Morita equivalence of semiring reducts implies Morita equivalence of the full algebras (via bimodules which in some appropriate sense preserve the additional structure).
This question was suggested to the author by Jose Meseguer at Logic and Applications 2025 in Dubrovník, Croatia.
The author has investigated this problem but removed results stated in an earlier version due to subtleties concerning homomorphism modules and tensor products pointed out by a diligent, anonymous referee.

This work is just the beginning of what could be a large research program.
One exciting direction is to develop homology and cohomology theory for Kleene algebras and investigate what implications the presence of nontrivial (co)homology classes have for computations represented in a given algebra.\\

\textbf{Formalization Statement:} Proposition~\ref{prop-no-kleene-right-adjoint}, lemma~\ref{KA-idempotent-corner}, theorems~\ref{KA-strong-morita}, \ref{KA-morita-expansion}, \ref{thm-kleene-module-equivalence}, and~\ref{thm-projective-automata-morita}, and the counterexample in subsection~\ref{sec-markings-bisimulations} have been formalized in Lean~4.19.0 using mathlib~v4.19.0.
The full source of the formalization, together with instructions for checking the proofs, is available at\\
 \url{https://github.com/lserafin17/kleene-morita/}.

\textbf{AI Disclosure Statement:} Variants of GPT-6 were used extensively to advance from the published RAMICS 2026 paper~\cite{serafin-morita} to this paper.
Preparation of the RAMICS paper did not involve any use of generative artificial intelligence, so it provides a useful comparison.
The large language model found an issue with the claimed tensor-hom adjunction for Kleene modules and was able to correct this issue with minimal additional input.
The model also introduced the author to the notion of corners of a semiring, which had been implicit in the earlier paper.
The author had the idea of showing that Morita-equivalent Kleene algebras have equivalent categories of automata, but the model uncovered issues with the conjectures which the author was trying to prove and made suggestions which
after some refinement in dialogue with the author led to the projective automata introduced in this paper.
The formalization was carried out mostly by GPT-6, with maybe two dozen substantive interactions with the author guiding how results are to be formalized, naming conventions, and precise statements of the results.
The author verified the correctness of the Lean formalization independently using his own installation of Lean which was in no way directly accessible to large language models.
On the whole the fundamental ideas are due to the author but the large language model has been immensely helpful in uncovering errors, reworking proofs, revising conjectures, and formally verifying most of the contents of the paper.
The author takes full responsibility for the content of this paper.

\bibliographystyle{plain}
\bibliography{refs}

@book{adamek-rosicky-locally-presentable,
  author = {Ad\'{a}mek, Ji\v{r}\'{i} and Rosick\'{y}, Ji\v{r}\'{i}},
  title = {Locally Presentable and Accessible Categories},
  series = {London Mathematical Society Lecture Note Series},
  volume = {189},
  publisher = {Cambridge University Press},
  year = {1994},
  doi = {10.1017/CBO9780511600579},
  url = {https://doi.org/10.1017/CBO9780511600579}
}

@article {banagl-tensor-semimodule,
    AUTHOR = {Banagl, Markus},
     TITLE = {The tensor product of function semimodules},
   JOURNAL = {Algebra Universalis},
  FJOURNAL = {Algebra Universalis},
    VOLUME = {70},
      YEAR = {2013},
    NUMBER = {3},
     PAGES = {213--226},
      ISSN = {0002-5240,1420-8911},
   MRCLASS = {16Y60},
  MRNUMBER = {3116477},
MRREVIEWER = {Evgeni\u i\ Mikha\u ilovich\ Vechtomov},
       DOI = {10.1007/s00012-013-0249-9},
       URL = {https://doi.org/10.1007/s00012-013-0249-9},
}

@article {bergman-coproducts-varieties,
    AUTHOR = {Bergman, George M.},
     TITLE = {On coproducts in varieties, quasivarieties and prevarieties},
   JOURNAL = {Algebra Number Theory},
  FJOURNAL = {Algebra \& Number Theory},
    VOLUME = {3},
      YEAR = {2009},
    NUMBER = {8},
     PAGES = {847--879},
      ISSN = {1937-0652,1944-7833},
   MRCLASS = {08B25 (03C05 08A60 08B26 08C15)},
  MRNUMBER = {2587406},
MRREVIEWER = {Keith\ A.\ Kearnes},
       DOI = {10.2140/ant.2009.3.847},
       URL = {https://doi.org/10.2140/ant.2009.3.847},
}

@article {brouwer-bicategorical-morita,
    AUTHOR = {Brouwer, R. M.},
     TITLE = {A bicategorical approach to {M}orita equivalence for von
              {N}eumann algebras},
   JOURNAL = {J. Math. Phys.},
  FJOURNAL = {Journal of Mathematical Physics},
    VOLUME = {44},
      YEAR = {2003},
    NUMBER = {5},
     PAGES = {2206--2214},
      ISSN = {0022-2488,1089-7658},
   MRCLASS = {46L10 (16D20 18D05 46M15 46N50)},
  MRNUMBER = {1972774},
MRREVIEWER = {Michael\ Frank},
       DOI = {10.1063/1.1563733},
       URL = {https://doi.org/10.1063/1.1563733},
}

@book {burris-sankappanavar,
    AUTHOR = {Burris, S. and Sankappanavar, H. P.},
     TITLE = {A course in universal algebra},
    SERIES = {Graduate Texts in Mathematics},
    VOLUME = {78},
 PUBLISHER = {Springer-Verlag, New York-Berlin},
      YEAR = {1981},
     PAGES = {xvi+276},
      ISBN = {0-387-90578-2},
   MRCLASS = {08-01},
  MRNUMBER = {648287},
MRREVIEWER = {R.\ S.\ Pierce},
}

@incollection {ehm-etal-kleene-modules,
    AUTHOR = {Ehm, T. and M\"oller, B. and Struth, G.},
     TITLE = {Kleene modules},
 BOOKTITLE = {Relational and {K}leene-algebraic methods in computer science},
    SERIES = {Lecture Notes in Comput. Sci.},
    VOLUME = {3051},
     PAGES = {112--123},
 PUBLISHER = {Springer, Berlin},
      YEAR = {2004},
      ISBN = {3-540-22145-X},
   MRCLASS = {03G25 (68Q60)},
  MRNUMBER = {2166158},
       DOI = {10.1007/978-3-540-24771-5\_10},
       URL = {https://doi.org/10.1007/978-3-540-24771-5_10},
}

@book {eklund-etal-semigroups-complete-lattices,
    AUTHOR = {Eklund, P. and Guti\'errez Garc\'ia, J. and H\"ohle,
              U. and Kortelainen, J.},
     TITLE = {Semigroups in complete lattices},
    SERIES = {Developments in Mathematics},
    VOLUME = {54},
    PUBLISHER = {Springer, Cham},
      YEAR = {2018},
     PAGES = {xxi+326},
      ISBN = {978-3-319-78947-7; 978-3-319-78948-4},
   MRCLASS = {06-02 (06F07 18C20 18D15 20M18 46L05 47-02)},
  MRNUMBER = {3793185},
MRREVIEWER = {Ittay\ Weiss},
       DOI = {10.1007/978-3-319-78948-4},
       URL = {https://doi.org/10.1007/978-3-319-78948-4},
}

@book {golan-semirings,
    AUTHOR = {Golan, J. S.},
     TITLE = {Semirings and their applications},
 PUBLISHER = {Kluwer Academic Publishers, Dordrecht},
      YEAR = {1999},
     PAGES = {xii+381},
      ISBN = {0-7923-5786-8},
   MRCLASS = {16Y60 (68Q01)},
  MRNUMBER = {1746739},
MRREVIEWER = {Udo\ Hebisch},
       DOI = {10.1007/978-94-015-9333-5},
       URL = {https://doi.org/10.1007/978-94-015-9333-5},
}

@book {hebisch-weinert-semirings-cs,
    AUTHOR = {Hebisch, U. and Weinert, H. J.},
     TITLE = {Semirings: algebraic theory and applications in computer
              science},
    SERIES = {Series in Algebra},
    VOLUME = {5},
      NOTE = {Translated from the 1993 German original},
 PUBLISHER = {World Scientific Publishing Co., Inc., River Edge, NJ},
      YEAR = {1998},
     PAGES = {viii+361},
      ISBN = {981-02-3601-8},
   MRCLASS = {16Y60 (68Q70)},
  MRNUMBER = {1704233},
MRREVIEWER = {C.\ J.\ Maxson},
       DOI = {10.1142/3903},
       URL = {https://doi.org/10.1142/3903},
}

@incollection {hopkins-leiss,
    AUTHOR = {Hopkins, Mark and Lei\ss, Hans},
     TITLE = {Normal forms for elements of the {$^*$}-continuous {K}leene
              algebras {$K\otimes_{\mathcal{R}}C_2'$}},
 BOOKTITLE = {Relational and algebraic methods in computer science},
    SERIES = {Lecture Notes in Comput. Sci.},
    VOLUME = {13896},
     PAGES = {122--139},
 PUBLISHER = {Springer, Cham},
      YEAR = {2023},
      ISBN = {978-3-031-28082-5; 978-3-031-28083-2},
   MRCLASS = {68Q70},
  MRNUMBER = {4590819},
       DOI = {10.1007/978-3-031-28083-2\_8},
       URL = {https://doi.org/10.1007/978-3-031-28083-2_8},
}

@article {katsov-injective-envelopes,
    AUTHOR = {Katsov, Y.},
     TITLE = {Tensor products and injective envelopes of semimodules over
              additively regular semirings},
   JOURNAL = {Algebra Colloq.},
  FJOURNAL = {Algebra Colloquium},
    VOLUME = {4},
      YEAR = {1997},
    NUMBER = {2},
     PAGES = {121--131},
      ISSN = {1005-3867,0219-1733},
   MRCLASS = {16Y60 (08B30 18A40 18G05)},
  MRNUMBER = {1682765},
MRREVIEWER = {Juturu\ Hanumanthachari},
}

@article {katsov-nam-morita,
    AUTHOR = {Katsov, Y. and Nam, T. G.},
     TITLE = {Morita equivalence and homological characterization of
              semirings},
   JOURNAL = {J. Algebra Appl.},
  FJOURNAL = {Journal of Algebra and its Applications},
    VOLUME = {10},
      YEAR = {2011},
    NUMBER = {3},
     PAGES = {445--473},
      ISSN = {0219-4988,1793-6829},
   MRCLASS = {16Y60},
  MRNUMBER = {2812121},
MRREVIEWER = {Maria\ Sanziana\ Pop},
       DOI = {10.1142/S0219498811004793},
       URL = {https://doi.org/10.1142/S0219498811004793},
}

@incollection{kozen-action,
  title={On action algebras},
  author={Kozen, D.},
  booktitle={Logic and Information Flow},
  editor={van Eijck, J. and Visser, A.},
  publisher={MIT Press},
  pages={78--88},
  year={1994}
}

@article{kozen-completeness,
  author={Kozen, D.},
  title={A completeness theorem for {Kleene} algebras and the algebra of regular events},
  journal={Information and Computation},
  volume={110},
  number={2},
  pages={366--390},
  year={1994},
  publisher={Elsevier}
}

@article {kozen-parikh-pdl,
    AUTHOR = {Kozen, D. and Parikh, R.},
     TITLE = {An elementary proof of the completeness of {PDL}},
   JOURNAL = {Theoret. Comput. Sci.},
  FJOURNAL = {Theoretical Computer Science},
    VOLUME = {14},
      YEAR = {1981},
    NUMBER = {1},
     PAGES = {113--118},
      ISSN = {0304-3975,1879-2294},
   MRCLASS = {03B45 (68B10)},
  MRNUMBER = {609516},
MRREVIEWER = {G.\ E.\ Mints},
       DOI = {10.1016/0304-3975(81)90019-0},
       URL = {https://doi.org/10.1016/0304-3975(81)90019-0},
}

@incollection {kozen-frederick-KAT,
    AUTHOR = {Kozen, D. and Smith, F.},
     TITLE = {Kleene algebra with tests: completeness and decidability},
 BOOKTITLE = {Computer science logic ({U}trecht, 1996)},
    SERIES = {Lecture Notes in Comput. Sci.},
    VOLUME = {1258},
     PAGES = {244--259},
 PUBLISHER = {Springer, Berlin},
      YEAR = {1997},
      ISBN = {3-540-63172-0},
   MRCLASS = {03B70 (68Q60)},
  MRNUMBER = {1611502},
       DOI = {10.1007/3-540-63172-0\_43},
       URL = {https://doi.org/10.1007/3-540-63172-0_43},
}

@incollection {kozen-ka-cl-semirings,
    AUTHOR = {Kozen, D.},
     TITLE = {On {K}leene algebras and closed semirings},
 BOOKTITLE = {Mathematical foundations of computer science ({B}ansk\'a{}
              {B}ystrica, 1990)},
    SERIES = {Lecture Notes in Comput. Sci.},
    VOLUME = {452},
     PAGES = {26--47},
 PUBLISHER = {Springer, Berlin},
      YEAR = {1990},
      ISBN = {3-540-52953-5},
   MRCLASS = {68Q70 (03G25 16Y60 68Q55)},
  MRNUMBER = {1084822},
MRREVIEWER = {Xiang\ Li},
       DOI = {10.1007/BFb0029594},
       URL = {https://doi.org/10.1007/BFb0029594},
}

@book {lam-lectures,
    AUTHOR = {Lam, T. Y.},
     TITLE = {Lectures on modules and rings},
    SERIES = {Graduate Texts in Mathematics},
    VOLUME = {189},
 PUBLISHER = {Springer-Verlag, New York},
      YEAR = {1999},
     PAGES = {xxiv+557},
      ISBN = {0-387-98428-3},
   MRCLASS = {16-01},
  MRNUMBER = {1653294},
MRREVIEWER = {Jonathan\ Golan},
       DOI = {10.1007/978-1-4612-0525-8},
       URL = {https://doi.org/10.1007/978-1-4612-0525-8},
}

@book {maclane-working,
    AUTHOR = {Mac Lane, Saunders},
     TITLE = {Categories for the working mathematician},
    SERIES = {Graduate Texts in Mathematics},
    VOLUME = {5},
 PUBLISHER = {Springer-Verlag, New York-Berlin},
      YEAR = {1971},
     PAGES = {ix+262},
   MRCLASS = {18-02},
  MRNUMBER = {354798},
MRREVIEWER = {H.-B.\ Brinkmann},
}

@inproceedings{serafin-morita,
  title={Morita Rigidity for Kleene Algebras},
  author={Serafin, Luke},
  booktitle={International Conference on Relational and Algebraic Methods in Computer Science},
  pages={301--317},
  year={2026},
  organization={Springer}
}

@article{katsov-nam-radicals,
  author = {Katsov, Y. and Nam, T. G.},
  title = {On Radicals of Semirings and Related Problems},
  journal = {Communications in Algebra},
  volume = {42},
  number = {12},
  pages = {5065--5099},
  year = {2014},
  doi = {10.1080/00927872.2013.833208},
  note = {Preprint: arXiv:1206.1580v1 [math.RA], 2012},
  url = {https://doi.org/10.1080/00927872.2013.833208}
}

@techreport{kozen-typed,
  author = {Kozen, Dexter},
  title = {Typed {Kleene} Algebra},
  institution = {Computer Science Department, Cornell University},
  number = {TR98-1669},
  year = {1998},
  month = mar,
  url = {https://hdl.handle.net/1813/7323}
}

@techreport{moeller-typed,
  author = {M{\"o}ller, Bernhard},
  title = {Typed {Kleene} Algebras},
  institution = {Institut f{\"u}r Informatik, Universit{\"a}t Augsburg},
  number = {1999-08},
  year = {1999},
  url = {https://opus.bibliothek.uni-augsburg.de/opus4/262}
}

@article{pous-untyping,
  author = {Pous, Damien},
  title = {Untyping Typed Algebras and Colouring Cyclic Linear Logic},
  journal = {Logical Methods in Computer Science},
  volume = {8},
  number = {2},
  pages = {13:1--13:21},
  year = {2012},
  doi = {10.2168/LMCS-8(2:13)2012},
  url = {https://doi.org/10.2168/LMCS-8(2:13)2012}
}
\end{document}